\documentclass{amsart}

\usepackage{latexsym}
\usepackage{linearb}
\usepackage[LGR, OT1]{fontenc}
\usepackage{amsmath,amsthm}
\usepackage{amssymb}
\usepackage{upgreek}
\usepackage[greek,english]{babel}
\usepackage{teubner}
\usepackage{MnSymbol, wasysym}
\usepackage{graphicx}
\usepackage{marvosym}
\usepackage[margin=1.1in,dvips]{geometry}
\usepackage{color}
\usepackage{enumerate}
\usepackage{comment}
\usepackage{mathrsfs}
\usepackage{mathtools}

\usepackage{bbold}

\usepackage{bbm}

\usepackage{multicol}

\usepackage[toc,page]{appendix}

\usepackage{physics}

\usepackage{hyperref}
\hypersetup{%
  colorlinks = true,
  linkcolor  = black,
  citecolor = blue
}

\usepackage{slashed}

\usepackage{xcolor}

\usepackage[textwidth=20mm,bordercolor=purple]{todonotes}

\newtheorem{innercustomTheorem}{Theorem}
\newenvironment{customTheorem}[1]
  {\renewcommand\theinnercustomTheorem{#1}\innercustomTheorem}
  {\endinnercustomTheorem}

\newtheorem{remark}{Remark}[section]

\newtheorem{proposition}{Proposition}[section]

\newtheorem*{claim*}{Claim}

\newtheorem*{rough version}{Rough Version}

\newtheorem*{theorem*}{Theorem}

\newtheorem*{corollary*}{Corollary}

\newenvironment{sketch proof}{\proof}{\endproof}

\makeindex

\numberwithin{equation}{section}

\title[ local integrated decay estimates for spherically symmetric black holes]{{a note on integrated local energy decay estimates \\ for spherically symmetric black hole spacetimes}}

\author{Gustav Holzegel} 
\address{Imperial College London, Department of Mathematics,
South Kensington Campus, London SW7 2AZ, UK}
\email{g.holzegel@imperial.ac.uk}
\address{Universität Münster,
Mathematisches Institut, Einsteinstrasse 62, 48149 Münster, Germany}
\email{gholzege@uni-muenster.de}

\author{Georgios Mavrogiannis}
\address{Rutgers University, Department of Mathematics, New Brunswick, NJ 08903 USA}
\email{gm758@math.rutgers.edu}

\author{Renato Velozo Ruiz}
\address{University of Toronto, Department of Mathematics, 40 St. George Street, Toronto, ON, Canada}
\email{renato.velozo.ruiz@utoronto.ca}

\date\today

\begin{document}

\begin{abstract}
We present short proofs of integrated local energy decay estimates on Schwarzschild, extremal Reissner--Nordstr\"om,  and Schwarzschild--de~Sitter spacetimes. The proofs employ novel global physical space multipliers, which, besides their remarkable simplicity, (a) are directly derivable from the geodesic flow, (b) do not require decomposition into spherical harmonics, and (c) whose boundary terms can be controlled by the conserved $T$-energy alone. We also elaborate on the intimate connection between the multipliers of the present paper and the globally good commutators introduced in~\cite{gustav,mavrogiannis}. 
\end{abstract}


\maketitle

{
  \hypersetup{linkcolor=black}
}

\section{Introduction}\label{sec: intro}

Integrated local energy decay estimates for hyperbolic equations on the exterior of black hole backgrounds have been studied intensely in the past two decades, most often in connection with the stability problem for black holes \cite{DR4, BlueSoffer1, AnderssonBlue, tatarutohaneanuKerr, DR2, Aretakis, SRTdC2}. The importance of such estimates lies in their global character, as they necessarily capture the geometry of the entire spacetime, including the well-known phenomenon of trapped geodesics and its associated degeneration in the estimate. Moreover, once such an estimate has been established, it can be combined with well-known estimates near the horizon (going back to \cite{DR4}) and near the asymptotically flat end \cite{DR6} to prove inverse polynomial decay estimates and more refined decay statements including, with additional work, sharp decay rates \cite{AAGsph, AAGextRN, Hin, AAGkerr}. 

In this paper, we revisit some of the simplest settings where an integrated local energy decay estimate can be proven:  The covariant wave equation $\Box_{g} \psi =0$ on $(\mathcal{M},g)$ the exterior of the Schwarzschild, the extremal Reissner--Nordström, and the Schwarzschild--de Sitter black hole geometry. Note that being static, these examples admit a coercive energy conservation law from the timelike Killing field. Our main motivation here is not so much to produce new results, although some of the estimates we state are new,  but rather to 
\begin{enumerate}
	\item  Streamline existing proofs by providing simple global multipliers leading to very short proofs. 
	\item Derive the exact form of these multipliers from global considerations regarding geodesic flow.
\end{enumerate}

For the Schwarzschild geometry, the proof of such estimates goes back to~\cite{DR4} and~\cite{BlueSoffer1}, who provided the first, rather elaborated, constructions of multipliers leading to integrated local energy decay estimates. However, in both works, the multipliers depended on the angular momentum number $\ell$ (of a decomposition of the solution into spherical harmonics), which complicated direct non-linear applications. This drawback was resolved in~\cite{DR7} at the cost of commuting with angular momentum operators, i.e.~invoking a higher order energy. Finally,~\cite{MMTT-stricharz} provided a construction of a single multiplier that worked for all frequencies without any commutation. Their construction relied on the redshift vector field of~\cite{DR4}: a small amount of it had to be added in the construction to compensate for terms of the wrong sign near the horizon. In particular, the estimate of~\cite{MMTT-stricharz} does not provide an integrated local energy that can be controlled by the $T$-energy arising from the Killing field. In contrast, our Theorem \ref{theorem 1} below controls an integrated decay norm in terms of the uncommuted initial $T$-energy alone. Remarkably, its proof only invokes multiplying the covariant wave equation (\ref{renow}) by $X\psi$ where $X$ is (expressed in standard Schwarzschild coordinates)
\begin{align} \label{XinS}
X= \left(1-\frac{2M}{r}\right) f(r) \partial_r + 2\left(1-\frac{2M}{r}\right) (\partial_r f) (r),  \ \ \ \textrm{where} \ f(r)= \left(1-\frac{3M}{r}\right) \sqrt{1+\frac{6M}{r}}.
\end{align}
While the degenerations at the horizon and the trapped set in (\ref{XinS}) are well-known, the factor of $\sqrt{1+\frac{6M}{r}}$ may look rather artificial. However, the exact form of $f$ is obtained when writing down the radial momentum $p^r$ (as a function of $r$) corresponding to a future-trapped null geodesic, i.e.~a null geodesic approaching $r=3M$ asymptotically.  We note that the form of $f$ also appeared in the commutator vector fields introduced in~\cite{gustav,mavrogiannis}. We elaborate on this connection in Section \ref{sec:nohyp} below. 

The unifying principle of ``deriving” the global form of the multiplier from the parametrisation of trapped geodesics in phase space is at the heart of this paper. It is remarkable and still slightly mysterious to us that the multipliers thus obtained provide direct and simple proofs of integrated local energy decay estimates in many classical cases: The Schwarzschild geometry in higher dimensions (here our Theorem \ref{theorem 1.1}, which concerns $n=4+1$ dimensions, considerably simplifies the constructions in \cite{LaulMetcalfe,volker-HighSchwarzschild}), the extremal Reissner--Nordström geometry (here our Theorem \ref{theorem 3} simplifies the construction in \cite{Aretakis}) and the Schwarzschild--de Sitter geometry (our Theorem \ref{theorem 2} below, which concerns the conformal wave equation).

It is of course a natural question whether the above procedure to construct multipliers leads to novel microlocal multipliers and a simplification of~\cite{DR2} in the Kerr case. Unfortunately, the constructions and the computations become significantly more involved. In order not to blur the simplicity of the construction in the basic examples, we postpone an investigation of this to the future.

{\bf Acknowledgments.}
G.H.~acknowledges support by the Alexander von Humboldt Foundation in the framework of the Alexander von Humboldt Professorship endowed by the Federal Ministry of Education and Research, ERC Consolidator Grant 772249 and funding through Germany’s Excellence Strategy EXC 2044 390685587, Mathematics Münster: Dynamics--Geometry--Structure. G.M. would like to thank Hans Lindblad for insightful conversations.

\section{Preliminaries}

\subsection{The spacetimes} \label{sec:spacetime}
We study the (static, spherically symmetric) black hole exteriors corresponding to the Schwarzschild, the extremal Reissner--Nordstr\"om and the Schwarzschild--de Sitter metric. We let $(\mathcal{M},g)$ be an $n$-dimensional Lorentzian manifold $(\mathcal{M},g)$ with $\mathcal{M}=\mathbb{R}_t \times (\rho_1,\rho_2)_r \times \mathbb{S}^{n-2}$ and
\begin{equation} \label{gform}
	g= -\xi(r) dt\otimes dt +(\xi(r))^{-1}dr\otimes dr +r^2 d\sigma_{\mathbb{S}^{n-2}} \, , 
\end{equation} 
where $\xi : (\rho_1,\rho_2) \rightarrow \mathbb{R}^+$ is smooth, and $d\sigma_{\mathbb{S}^{n-2}}$ denotes the round metric on the unit sphere $\mathbb{S}^{n-2}$.

Specifically, given $M>0$, we define

\begin{enumerate}
\item $\mathrm{Schw}^{1+3}:=\mathbb{R}_t\times (2M,\infty)_r\times \mathbb{S}^2$ with $\xi(r) = 1-\frac{2M}{r}$ and $n=4$ in (\ref{gform}), 
\item $\mathrm{Schw^{1+4}}:=\mathbb{R}_t\times (\sqrt{2M},\infty)_r\times \mathbb{S}^3$ with $\xi(r)=1-\frac{{2M}}{r^{2}}$ and $n=5$ in (\ref{gform}), 
\item $\mathrm{ERN}^{1+3}:=\mathbb{R}_t\times (M,\infty)_r\times \mathbb{S}^2$ with $\xi(r)=(1-\frac{M}{r})^2$ and $n=4$ in (\ref{gform}). 
\end{enumerate}
Given also $L>0$ and~$M>0$ that satisfy the subextremality condition~$0<\frac{M^2}{L^2}<\frac{1}{27}$ and letting $\bar{r}_+>r_+>0$ denote the roots of $\xi(r)=1- \frac{2M}{r}-\frac{r^2}{L^2}$:
\begin{enumerate}
\setcounter{enumi}{3}
\item $\mathrm{SchwdS^{1+3}}:=\mathbb{R}_t\times (r_+,\bar{r}_+)_r\times \mathbb{S}^2$ with $\xi(r)=1- \frac{2M}{r}-\frac{r^2}{L^2}$ and $n=4$ in (\ref{gform}).
\end{enumerate}
We will often employ the well-known tortoise coordinate $r^\star \in (-\infty,\infty)$ defined through $\frac{dr^\star}{dr} = \frac{1}{\xi(r)}$. 

Finally, note that the natural spacetime volume form associated with the metrics (\ref{gform}) can be written as~$dvol_{\mathcal{M}} = \xi r^{n-2} dt\wedge dr^\star\wedge dvol_{\mathbb{S}^{n-2}}$, where $ dvol_{\mathbb{S}^{n-2}}$ denotes the standard volume form on~$\mathbb{S}^{n-2}$. 

\subsection{The spacetime foliations} \label{sec:foliation}

For each of the spacetimes of Section~\ref{sec:spacetime} above we consider in addition a foliation by spherically symmetric, smooth spacelike slices $\Sigma_\tau$, which are defined as the push-forward~$\Sigma_\tau= \varphi_{\tau}(\Sigma_0)$ associated with the integral curves of the Killing vector field~$\partial_t$, where for $(1)$--$(3)$ the leaf~$\Sigma_0$ is chosen to connect the future event horizon with future null infinity, and for $(4)$ the leaf $\Sigma_0$ is chosen to connect the future event horizon and the future cosmological horizon. The precise form of the slices is inessential in what follows; moreover, in the cases $(1)$--$(3)$ all of our arguments would equally work for slices ending at spatial infinity. The region between two slices $\Sigma_{\tau_1}$, $\Sigma_{\tau_2}$ will be denoted $\mathcal{M}\left(\tau_1,\tau_2\right)$.

\subsection{The covariant wave equation} \label{sec:covariantwave}
We will consider smooth solutions to the covariant wave equation 
\begin{align} \label{kgwave}
\Box_g \phi - \mu \phi= 0 \, 
\end{align} 
in the following four cases:
\begin{itemize}
\item $\mu = 0$ in (\ref{kgwave}) and $(\mathcal{M},g)$ being one of the spacetimes in $(1)$--$(3)$.
\item $\mu =\frac{2}{L^2}$  in (\ref{kgwave}) and $(\mathcal{M},g)$ being the spacetime $(4)$. 
\end{itemize}
The last case is the conformal wave equation, the natural analogue of the ``massless" case for non-vanishing cosmological constant. 
Setting $\psi = \phi r^{\frac{n-2}{2}}$, we can write (\ref{kgwave}) as
\begin{align} \label{renow}
-\partial_t^2 \psi + \partial_{r^\star}^2 \psi + \frac{\xi(r)}{r^2} \Delta_{\mathbb{S}^{n-2}} \psi -V(r) \psi - \mu \xi(r) \psi= 0 \ ,  \ \ \ \textrm{where} \ V(r) = r^{-\frac{n-2}{2}} \left( r^{\frac{n-2}{2}}\right)^{\prime \prime}.
\end{align}
In the above, a prime denotes differentiation with respect to $r^\star$ and $\Delta_{\mathbb{S}^{n-2}}$ is the spherical Laplacian on the unit sphere. It will often be algebraically simpler to work with (\ref{renow}) directly instead of (\ref{kgwave}). Note that in all four cases considered there is a $c$ depending only on the black hole parameters such that
\begin{align} \label{Vpos}
V+\mu \xi \geq \frac{c}{r^3}\xi \ \ \ \textrm{holds for $r \in (\rho_1,\rho_2)$.} 
\end{align}

\subsection{The basic energy currents} 
Using notation of \cite{dafermos2022quasilinear}, Section 3.3, we define from a tuple $(X,w,q,\varpi)$ consisting of a vectorfield $X$, a function $w$, a one-form $q$ and an $(n-2)$-form $\varpi$ on $\mathcal{M}$, the currents 
\begin{align} 
J_\mu^{X, w, q, \varpi} [\phi] &:= T_{\mu \nu} [\phi] X^\nu + w \phi \partial_\mu \phi + q_\mu \phi^2 + \star d(\phi^2 \varpi)_\mu  \, , \\
K^{X,w,q} [\phi] &:= T^{\mu \nu} \pi^X_{\mu \nu} + \nabla^\mu w \phi \nabla_\mu \phi + w \left( \nabla^\mu \phi \nabla_\mu \phi + \mu \phi^2\right) + \nabla^\mu q_\mu \phi^2 + 2 \phi q_\mu g^{\mu \nu} \partial_\nu \phi \, ,
\end{align}
where 
\begin{align}
T_{\mu \nu}[\phi] = \partial_\mu \phi \partial_\nu \phi - \frac{1}{2} g_{\mu \nu} \left( g^{\alpha \beta} \partial_\alpha \phi \partial_\beta \phi+ \mu \phi^2 \right)  \, ,  \ \ \ \ \ \ \pi_{\mu \nu}^X = \frac{1}{2} (\mathcal{L}_X g)_{\mu \nu} = \frac{1}{2} \left(\nabla_\mu X_\nu +\nabla_\nu X_\mu \right) .
\end{align}
One easily checks that for $\phi$ satisfying (\ref{kgwave}), the above currents are related by the divergence identity
\begin{align} \label{divid}
\nabla^\mu J_\mu^{X, w, q, \varpi} [\phi] = K^{X,w,q} [\phi] \, .
\end{align}
Below we will consider two currents $J^{(1)}=J^{X_1,w_1,q_1,\varpi_1}[\phi]$ and $J^{(2)}=J^{X_2,w_2,q_2,\varpi_2}[\phi]$. 

The current $J^{(1)}$ arises from the static Killing field $T=\partial_t$ and is defined by 
\begin{align} \label{Tdata}
X_1 = \partial_t \ \ \ , \ \ \ w_1= 0 \ \ \ , \ \ \ q_1= 0 \ \ \ , \ \ \  \varpi_1 =  \left(-1\right)^{n+1} \frac{n-2}{4} \frac{\xi}{r} r^{n-2} d\textrm{vol}_{\mathbb{S}^{n-2}} \, .
\end{align}
Clearly $K^{X_1,w_1,q_1}[\phi]=0$ and one computes the components (expressed in terms of $\psi=r^{\frac{n-2}{2}} \phi$)
\begin{align}
J^{(1)}_t  &= \frac{1}{2} \frac{1}{r^{n-2}} \left((\partial_t \psi)^2 + (\partial_{r^\star} \psi)^2 +\frac{\xi(r)}{r^2} |\nabla_{\mathbb{S}^{n-2}} \psi|^2 +\left(V(r) + \mu \xi(r) \right)\psi^2\right) \nonumber \, ,  \\  J^{(1)}_{r^\star} &=  \frac{1}{r^{n-2}}  \left(\partial_t \psi \partial_{r^\star} \psi\right) \, , 
\label{Tcomponents} \\ J^{(1)}_A &=  \frac{1}{r^{n-2}}  \left(\partial_t \psi \partial_{A} \psi\right) . \nonumber 
\end{align}

\begin{remark}
The addition of the non-trivial $(n-2)$-form in (\ref{Tdata}) produces the additional coercive zeroth order term $V\psi^2$ in the current $J^{(1)}_t$ when written in terms of $\psi$ instead of $\phi$, which is convenient in the analysis. For $\varpi_1=0$ one would obtain the usual $J^{(1)}_t = \frac{1}{2} \left((\partial_t \phi)^2 + (\partial_{r^\star} \phi)^2 +\frac{\xi(r)}{r^2} |\nabla_{\mathbb{S}^{n-2}} \phi|^2 +\mu \xi(r)\phi^2\right)$.
\end{remark}

The current $J^{(2)}=J^{X_2,w_2,q_2,\varpi_2}$ is the so-called Morawetz current, defined by 
\begin{align} \label{mocu}
X_2 = f\partial_{r^\star} \ \ \ , \ \ \ w_2=\frac{n-2}{2} \frac{f}{r}\xi + \frac{1}{2}f^\prime \ \ \ , \ \ \ (q_2)_\mu = -\frac{1}{2} \partial_\mu w_2 + \frac{n-2}{2} \frac{f^\prime}{r} \partial_\mu r \ \ \ , \ \ \  \varpi_2 =  0 \, ,
\end{align}
for a smooth radial function $f : (\rho_1,\rho_2) \rightarrow \mathbb{R}$ to be determined in (\ref{fde}). We compute its components
\begin{align}
J^{(2)}_t &=  \frac{1}{r^{n-2}} \left(f \partial_t \psi \partial_{r^\star} \psi + \frac{1}{2} f^\prime \partial_t \psi  \cdot \psi  \right)\, , \nonumber  \\  J^{(2)}_{r^\star} &=  \frac{1}{2} \frac{1}{r^{n-2}}  \left(f(\partial_{r^\star} \psi)^2 + f(\partial_{t} \psi)^2  - f\frac{\xi(r)}{r^2} |\nabla_{\mathbb{S}^{n-2}} \psi|^2 -(V+\mu \xi(r)) f \psi^2 + f^\prime \partial_{r^\star} \psi \cdot \psi -\frac{f^{\prime \prime} \psi^2}{2}\right) \, ,  \label{Mcomponents}
\\ J^{(2)}_A &=    \frac{1}{r^{n-2}}  \left(\partial_{r^\star} \psi \partial_{A} \psi\right) \, ,  \nonumber 
\end{align}
and the spacetime term $K^{(2)}=K^{X_2,w_2,q_2}[\phi]$:
\begin{align}
K^{(2)}= \frac{1}{r^{n-2}\xi} \left[ f^\prime  (\partial_{r^\star} \psi)^2  - \frac{f}{2}  \left( \frac{\xi(r)}{r^2}\right)^\prime |\nabla_{\mathbb{S}^{n-2}} \psi|^2- \left(\frac{1}{4} f^{\prime \prime \prime} + \frac{1}{2}(V+\mu \xi)^\prime f\right) \psi^2  \right] \, . \nonumber
\end{align}

We finally remark that the divergence identities associated with $J^{(1)}$ and $J^{(2)}$ can also be obtained directly ``by hand” if ones multiplies (\ref{renow}) by $\partial_t \psi$ and by $-f^\prime \psi -2 f\psi^\prime$ respectively.

\subsection{The basic energy}
We define the basic energy on the spacelike slices $\Sigma_\tau$ introduced in Section \ref{sec:foliation},
\begin{align} \label{ETdef}
\mathbb{E}^T[\psi](\tau):=\int_{\Sigma_\tau} J_\mu^{(1)} n^\mu_{\Sigma} \, .
\end{align}
Using (\ref{Vpos}) one finds that $\mathbb{E}^T[\psi](\tau)$ is coercive in each of the four cases considered, in fact we have 
\begin{align}
\mathbb{E}^T[\psi](\tau) \sim \int_{\Sigma_{\tau}} dr d\textrm{vol}_{\mathbb{S}^{n-2}}  \left[  (\partial_t \psi + \partial_{r^\star}\psi)^2  + \frac{ (\partial_t \psi - \partial_{r^\star}\psi)^2}{\xi(r) r^{1+\delta}}  +\frac{\xi(r)}{r^2} |\nabla_{\mathbb{S}^{n-2}} \psi|^2 +\frac{\xi}{r^3} \psi^2 \right] \, ,
\end{align}
where $0<\delta<1$ is determined from how the spacelike slices $\Sigma_\tau$ intersect null infinity and $\sim$ involves only constants depending on the black hole parameters. Applying the divergence identity (\ref{divid}) for the current $J^{(1)}_\mu$ yields for any $\tau_2 \geq \tau_1$ the estimate
\begin{align}
\mathbb{E}^T[\psi](\tau_2) \leq \mathbb{E}^T[\psi](\tau_1) \, .
\end{align}

\subsection{Deriving the multipliers from the geodesic flow} \label{sec:constructionprinciple}
Let $(\mathcal{M},g)$ be one of the spacetimes in $(1)$--$(4)$. The dynamics of a freely falling massless particle on $(\mathcal{M},g)$ takes place on the subset $\mathcal{P}$ of the tangent bundle $T\mathcal{M}$ given by
\begin{equation}\label{null_shell_relation}
	\mathcal{P}:=\Big\{(x,p)\in T\mathcal{M}: g_x(p,p) =0, \text{ where $p$ is future directed}\Big\}\, . 
\end{equation}
Here $(p^t, p^r, p^A)$ are the standard dual momentum coordinates in the tangent space. We consider the Hamiltonian flow set by the geodesic equations on $(\mathcal{M},g)$ given by 
\begin{equation*}
\frac{dx^{\alpha}}{ds}=p^{\alpha} ,\qquad \frac{dp^{\alpha}}{ds}=-\Gamma^{\alpha}_{\beta \gamma} (g)p^{\beta}p^{\gamma},
\end{equation*}
where $\Gamma^{\alpha}_{\beta \gamma} (g)$ denote the Christoffel symbols of $(\mathcal{M},g)$. The conserved quantities associated with staticity and spherical symmetry of the metric make the geodesic flow a \emph{completely integrable Hamiltonian flow in the sense of Liouville}. In particular, the conserved quantities can be exploited to reduce the radial geodesic motion to an ODE involving only the (conserved) \emph{particle energy} $E$ and the (conserved) \emph{total angular momentum} $l$ defined by 
 \begin{equation*}
E:=\xi \cdot p^t, \qquad l:=r\sqrt{g_{AB}p^A p^B} \, .
\end{equation*}
Indeed, by \eqref{null_shell_relation}, the particle energy $E$ can be written in terms of the radial coordinates by 
\begin{equation}\label{identity_particle_energy_angular_momentum}
	E^2=(p^r)^2+\mathcal{V}_l(r),\qquad \mathcal{V}_l(r):=\dfrac{l^2}{r^2}\xi,
\end{equation} 
where $\mathcal{V}_l(r)$ is the radial potential of the geodesic flow. Differentiating \eqref{identity_particle_energy_angular_momentum} along the flow, we obtain the radial geodesic equation 
\begin{equation}\label{radgeoeqn}
	\dfrac{dr}{ds}=p^r,\qquad \dfrac{dp^r}{ds}= \partial_r \mathcal{V}_l(r),
\end{equation} 
which defines the radial flow. Note that  for any of the spacetimes $(1)$--$(4)$, the radial potential $\mathcal{V}_l(r)$ has a unique critical point at $r=r_{trap}$ in $(\rho_1,\rho_2)$, which is seen to be a local maximum. It follows that $(r_{trap},0)$ is a hyperbolic fixed point for the radial flow which corresponds to the existence of null geodesics propagating in the photon sphere $\{r=r_{trap}\}$. We call \emph{trapped orbits} the geodesics contained in the so-called \emph{trapped set} 
\begin{equation} \label{Gammadef}
\Gamma:=\Big\{ (x,p)\in \mathcal{P}: r=r_{trap}, \quad p^r=0\Big\}.
\end{equation}
It can be shown that $\Gamma$ is \emph{normally hyperbolic}. See \cite{WZ11, DZ13, D15} for more information about the normal hyperbolicity of the trapped set in the Schwarzschild and Kerr black holes.

By the classical Hadamard--Perron theorem, the hyperbolic fixed point $(r_{trap},0)$ for the radial flow implies the existence of suitable stable and unstable manifolds in the $(r,p^r)$ plane. Equation \eqref{identity_particle_energy_angular_momentum} shows that along the stable manifolds, the conserved quantities must satisfy 
\begin{equation} \label{trarel}
\frac{l^2}{E^2}=\frac{r_{trap}^2}{\xi(r_{trap})}.
\end{equation}
This relation can be used to obtain an explicit characterisation of the stable and unstable manifolds associated to the radial flow. In phase space $\mathcal{P}$, the stable and unstable manifolds of the fixed point $(3M,0)$ correspond to the set of future-trapped and past-trapped geodesics, respectively. We say a geodesic is \emph{future-trapped} if $\lim_{s\to \infty }(r(s),p^r(s))=(r_{trap},0)$. Similarly, we say a geodesic is \emph{past-trapped} if $\lim_{s\to -\infty }(r(s),p^r(s))=(r_{trap},0)$. Using the relation \eqref{identity_particle_energy_angular_momentum}, we can write the conserved quantity $\frac{l^2}{E^2}$ as
\begin{equation}\label{quotient_conserved_quantities}
\frac{l^2}{E^2}-\frac{r_{trap}^2}{\xi(r_{trap})}=\frac{r^2}{\xi}\Big(1-\frac{\xi}{r^2}\frac{r_{trap}^2}{\xi(r_{trap})}\Big)-\frac{r^2}{\xi}\Big(\frac{p^r}{E}\Big)^2.
\end{equation}
It is now clear that for a geodesic to be future or past trapped, we need the relation (\ref{trarel}) to hold, so
\begin{align} \label{fde}
\frac{p^r}{E} = \pm \left( 1-\frac{\xi}{r^2}\frac{r_{trap}^2}{\xi(r_{trap})} \right)^\frac{1}{2} =: \pm f(r) \, .
\end{align}
The function $f$ defined by (\ref{fde}) is the $f$ to be used in the current (\ref{mocu}). Specifically, we obtain 
\begin{align}
f(r)&=\Big(1-\dfrac{3M}{r}\Big)\Big(1+\dfrac{6M}{r}\Big)^{\frac{1}{2}} \ & \ \ &\textrm{for Schwarzschild and $n=4$,}  \label{multipl_schw}
\\
f(r)&=\left(1-\frac{2 \sqrt{M}}{r}\right) \left(1+\frac{2 \sqrt{M}}{r}\right) \ & \ \ &\textrm{for Schwarzschild and $n=5$,} \label{multipl_schw5}
\\
f(r)&= \left(1-\frac{2 M}{r}\right)\Big(\frac{ r^2+4Mr -4 M^2}{ r^2}\Big)^{\frac{1}{2}} \ & \ \ &\textrm{for extremal Reissner--Nordstr\"om and $n=4$,} \label{multipl_ern}
\\
f (r)&= \frac{1}{\sqrt{1-27\frac{M^2}{L^2}}} \left(1-\frac{3M}{r}\right)\Big(1+\frac{6M}{r}\Big)^{\frac{1}{2}} \ &\ \ &\textrm{for Schwarzschild--de Sitter and $n=4$.} \label{multipl_sds}
\end{align} 
We abstain from introducing a further subscript to distinguish the different $f$’s as it will always be clear from the context which $f$ is to be used. 

We close this section by collecting an important monotonicity property: There exists a constant $c$ depending only on the black hole parameters such that
\begin{align} \label{fprimepos}
f^\prime \geq c \frac{\xi}{r^3}  \ \ \ \textrm{for (\ref{multipl_schw}), (\ref{multipl_schw5}), (\ref{multipl_sds}),}  \qquad \qquad  f^\prime \geq c  \left(1-\frac{M}{r}\right)\frac{\xi}{r^3} \ \ \ \textrm{for (\ref{multipl_ern}).} 
\end{align}

\section{The main results}\label{sec: main thms}

We recall the setting of Section \ref{sec:covariantwave} and the energy $ \mathbb{E}^T[\psi](\tau)$ defined in (\ref{ETdef}). We will state our Theorems for smooth solutions but by the usual density arguments the estimates also hold for appropriately defined weak solutions. When stating the results below we will use $A \lesssim B$ to mean $A \leq C B$ for a constant $C$ depending only on the black hole parameters.

\begin{customTheorem}{1}(Integrated decay estimate on Schwarzschild \cite{DR4})\label{theorem 1}
	Let~$\phi$ be a solution of the wave equation $\Box_g \phi = 0$ for $(\mathcal{M},g)$ the exterior of a Schwarzschild black hole~$\mathrm{Schw}^{1+3}$. Then $\psi = \phi r$ satisfies the estimate
	\begin{equation} \label{eq:onlyT}
				\int_{\mathcal{M}(0,\tau)} \frac{dvol_{\mathcal{M}}}{r^{2}} \left(\frac{1}{r^3}(\partial_{r^\star} \psi)^2+ 	\frac{1}{r^3}\left(1-\frac{3M}{r}\right)^2\left((\partial_t \psi)^2+|r\slashed{\nabla}\psi|^2\right) + \frac{\psi^2}{r^4}  \right)\lesssim  \mathbb{E}^T[\psi](0) \, .
	\end{equation}
\end{customTheorem}

Note that (\ref{eq:onlyT}) is equivalent to the estimate replacing $\psi$ by $\phi$ and $\frac{dvol_{\mathcal{M}}}{r^{2}}$ by $dvol_{\mathcal{M}}$. We state the estimate for $\psi$ simply because it is easiest to prove in this form. 
While the estimate (\ref{eq:onlyT}) has already been established in \cite{DR4}, we emphasise that our proof relies solely on  the simple multiplier (\ref{multipl_schw}), which is independent of the spherical harmonic number. The estimates of \cite{DR7, MMTT-stricharz} are obtained with multipliers independent of the spherical harmonic number but -- in comparison with (\ref{eq:onlyT}) -- either require higher order energies~\cite{DR7} or the current arising from a globally uniformly timelike vector field (instead of the Killing field $\partial_t$) on the right~(\cite{MMTT-stricharz}, see also~\cite{stogin}).

Theorem \ref{theorem 1} and its proof immediately generalise to the case of higher dimensional Schwarzschild metrics. We focus here on the $4+1$ dimensional case:

\begin{customTheorem}{2}\label{theorem 1.1}(Integrated decay estimate on 1+4 Schwarzschild)
	Let~$\phi$ be a solution to the wave equation $\Box_g \phi=0$ for $(\mathcal{M},g)$ the exterior of a~$1+4$ dimensional Schwarzschild black hole~$\mathrm{Schw}^{1+4}$. Then, $\psi=\phi r^\frac{3}{2}$ satisfies the estimate
	\begin{equation} \label{eq:hds}
		\int_{\mathcal{M}(0,\tau)} \frac{dvol_{\mathcal{M}}}{r^3} \left(\frac{1}{r^3}(\partial_{r^\star} \psi)^2+ 	\frac{1}{r^3}\left(1-\frac{2\sqrt{M}}{r}\right)^2\left((\partial_t \psi)^2+|r\slashed{\nabla}\psi|^2\right) + \frac{\psi^2}{r^4}  \right)\lesssim  \mathbb{E}^T[\psi](0).
	\end{equation}
\end{customTheorem}

The estimate (\ref{eq:hds}) in its exact form seems to be new. The paper \cite{LaulMetcalfe} obtains (\ref{eq:hds}) with the current arising from a uniformly timelike vector field on the right, while \cite{volker-HighSchwarzschild} requires higher order energies. 

Our next result addresses solutions of (\ref{kgwave}) on the extremal Reissner--Nordstr\"om geometry~\cite{Angelopoulos-Aretaki-Gajic,Giorgi-LinearstabilityofRN}.

\begin{customTheorem}{3}(Integrated decay estimate on extremal Reissner--Nordstr\"om)\label{theorem 3}
	Let~$\phi$ be a solution to the wave equation $\Box_g \phi =0$ the exterior of an extremal Reissner--Nordstr\"om black hole~$\mathrm{ERN}^{1+3}$. Then, $\psi = \phi r$ satisfies the estimate
	\begin{align} \label{erne}
		\int_{\mathcal{M}(0,\tau)} \frac{dvol_{\mathcal{M}}}{r^2} \left(1-\frac{M}{r}\right) \left(\frac{1}{r^3}(\partial_{r^\star} \psi)^2+ 	\frac{1}{r^3}\left(1-\frac{2{M}}{r}\right)^2\left((\partial_t \psi)^2+|r\slashed{\nabla}\psi|^2\right) + \frac{\psi^2}{r^4}  \right)\lesssim  \mathbb{E}^T[\psi](0).
		\end{align}
\end{customTheorem}

We note that the radial weights near infinity appearing in our main integrated estimates \eqref{eq:onlyT}--\eqref{erne} are non-optimal even if one insists on only using the $T$-energy $\mathbb{E}[\psi](0)$ on the right hand side, which one might want to use in scattering theory applications. It is well-known how to optimise them with a multiplier localised to infinity. As mentioned in the introduction,  if one is willing to use the redshift and weighted estimates near infinity, stronger estimates can be proven.

Our multiplier construction also produces an integrated local energy decay estimate for the conformal wave equation on the Schwarzschild--de Sitter exterior. In this case, there are two parameters involved and the resulting algebra is more involved. As a result, we only control the solution minus its spherical means. Note also that in this case, $r$-weights do not play any role as the exterior is compact in $r$.

\begin{customTheorem}{4}(Integrated decay estimate on Schwarzschild--de Sitter)\label{theorem 2}
	Let~$\phi$ be a solution to the conformal wave equation~$\Box\phi =\frac{2}{L^2}\phi$ on the exterior of a subextremal Schwarzschild--de~Sitter black hole~$\mathrm{SchwdS^{1+3}}$ which has vanishing spherical mean~$\frac{1}{vol(\mathbb{S}^{2})}\int d\sigma_{\mathbb{S}^{2}} \psi\equiv 0$. Then, $\psi = \phi r$ satisfies the estimate
	\begin{align} \label{sdst}
			\int_{\mathcal{M}(0,\tau)} dvol_{\mathcal{M}} \left((\partial_{r^\star} \psi)^2+ 	\left(1-\frac{3M}{r}\right)^2\left((\partial_t \psi)^2+|r\slashed{\nabla}\psi|^2\right) + \psi^2  \right)\lesssim  \mathbb{E}^T[\psi](0) \, .
\end{align}
\end{customTheorem}

\section{The proofs}\label{sec: proofs}

We first note that it is sufficient to prove the estimates (\ref{eq:onlyT})--(\ref{sdst}) without the $(\partial_t \psi)^2$-term on the left, which we will call the reduced estimate. Indeed, once the reduced estimate has been proven, a standard Lagrangian estimate recovers the missing derivative in terms of what has already been proven. 

To prove the reduced estimate, we will integrate over $\mathcal{M}(\tau_1,\tau_2)$ the divergence identity (\ref{divid}) associated with the the current $C_f \cdot J^{(1)} + J^{(2)}$ with  $f$ being the $f$ associated with the spacetime geometry under consideration and $C_f$ a constant to be fixed in Proposition \ref{lem:fproperties} below. We define
\begin{align}
\mathbb{E}^{aux}[\psi](\tau):= 2\int_{\Sigma_\tau} \left(C_f \cdot J_\mu^{(1)} + J_\mu^{(2)}\right)n^\mu_{\Sigma_{\tau}} \, ,
\end{align}
\begin{align} \label{idt}
\mathbb{I} \left[\psi, f \right] (\tau_1,\tau_2):= \int_{\mathcal{M}(\tau_1,\tau_2)} dt dr^\star d\sigma_{\mathbb{S}^{n-2}} \left[ 2f^\prime |\psi^{\prime}|^2  - f \left(\frac{\xi}{r^2}\right)^\prime |\nabla_{\mathbb{S}^{n-2}} \psi|^2 -\left(f (V+\mu\xi)^\prime +\frac{f^{\prime\prime\prime}}{2}\right)|\psi|^2\right].
\end{align}

We have the following statement:

\begin{proposition} \label{lem:fproperties}
Let $\phi$ be a solution of $\Box_g \phi= 0$ for $(\mathcal{M},g)$ one of the spacetimes (1)--(3) above or a solution of $\Box_g \phi - \frac{2}{L^2}\phi= 0$ for the spacetime (4). Let the radial function $f$ in the current $J^{(2)}$ satisfy 
\begin{align} \label{fconditions}
|f| + \bigg|\frac{r^2 f^\prime}{\xi}\bigg| + \bigg|\frac{r^3}{\xi} f^{\prime \prime}\bigg| \leq C \, , 
\end{align}
for some $C$ depending only on the black hole parameters. Then for any $\tau_2,\tau_1 \geq 0$
\begin{align} \label{figo}
\mathbb{I} \left[\psi, f \right] (\tau_1,\tau_2) \lesssim \mathbb{E}^T[\psi](\tau_1)  \, .
\end{align}
\end{proposition}

\begin{proof}
We recall (\ref{Vpos}) holds in all four cases considered. We claim that we can choose a $C_f$ depending only on parameters and the $C$ in (\ref{fconditions}) such that both 
\begin{equation}
\begin{split} \label{cubo}
c_1 J^{(1)}_\mu n_{\Sigma_\tau}^\mu &\leq C_f J^{(1)}_\mu n_{\Sigma_\tau}^\mu + J^{(2)}_\mu n_{\Sigma_\tau}^\mu \phantom{xx} \leq C_2 J^{(1)}_\mu n_{\Sigma_\tau}^\mu \, , \\
c_1 J^{(1)}_\mu (\partial_t)^\mu &\leq C_f J^{(1)}_\mu (\partial_t)^\mu + J^{(2)}_\mu (\partial_t)^\mu \leq C_2 J^{(1)}_\mu (\partial_t)^\mu  \, , 
\end{split}
\end{equation}
hold for constants $c_1, C_2$ depending only on the parameters. Indeed, this follows easily from the form of the currents (\ref{Tcomponents}) and (\ref{Mcomponents}) using (\ref{fconditions}) and the bound (\ref{Vpos}). We next apply the divergence identity (\ref{divid}) for the current $C_f \cdot J^{(1)} + J^{(2)}$ over $\mathcal{M}\left(\tau_1,\tau_2\right)$. It follows that
\[
\mathbb{I} \left[\psi, f \right] (\tau_1,\tau_2) \leq C_f \mathbb{E}^{aux}[\psi](\tau_1) 
\]
as all boundary terms in the divergence identity except the one on $\Sigma_{\tau_1}$ have favourable signs. The right hand side is now immediately converted to $ \mathbb{E}^{T}[\psi](\tau_1)$ using the bounds (\ref{cubo}). 
\end{proof}

Note that all $f$'s in (\ref{multipl_schw})--(\ref{multipl_sds}) satisfy the assumption (\ref{fconditions})  hence (\ref{figo}) indeed holds. We next recall the form of (\ref{idt}), which is the left hand side of (\ref{figo}).

Denoting the spherical average of $\psi$ by~$	\psi_0:=\frac{1}{vol(\mathbb{S}^{n-2})}\int dvol_{\mathbb{S}^{n-2}} \psi$ and letting $g : [\rho_1,\rho_2) \rightarrow \mathbb{R}$ be a smooth function (with $g \cdot r$ uniformly bounded) to be chosen, we can write (\ref{idt}) as
\begin{equation}
	\begin{aligned}
		&	\mathbb{I} \left[\psi, f \right] (\tau_1,\tau_2) 	:= \\ 	&\int_{\mathcal{M}(\tau_1,\tau_2)}  dt dr^\star d\sigma_{\mathbb{S}^{n-2}} \Bigg( 2f^\prime |(\psi-\psi_0)^{\prime}|^2  - f \left(\frac{\xi(r)}{r^2}\right)^\prime |\nabla_{\mathbb{S}^{n-2}} (\psi-{\psi}_0)|^2 -\left(f (V+\mu\xi)^\prime +\frac{f^{\prime\prime\prime}}{2}\right)|\psi-{\psi}_0|^2  \Bigg) \nonumber \\
		 &	\qquad +\int_{\mathcal{M}(\tau_1,\tau_2)} dt dr^\star d\sigma_{\mathbb{S}^{n-2}}  \Bigg( 2f^\prime |(\psi_0)^{\prime}-g\psi_0 |^2-\left(2f^\prime g^2 -2(f^\prime g)^\prime -  f (V+\mu\xi)^\prime +\frac{f^{\prime\prime\prime}}{2}\right)|{\psi}_0|^2 \Bigg) \, . \nonumber 
	\end{aligned}
\end{equation}
Recall now the positivity property (\ref{fprimepos}) and observe that all $f$ in (\ref{multipl_schw})--(\ref{multipl_sds}) satisfy~$-f(r)\left(\frac{\xi(r)}{r^2}\right)^\prime\geq 0$. (Specifically, recall that~$f(r)$ and~$\left(\frac{\xi(r)}{r^2}\right)^\prime$ both change sign at trapping.) Therefore, denoting by $ \lambda^{min}_{\mathbb{S}^{n-2}}$ the smallest non-trivial eigenvalue of the (negative) Laplacian on the unit-sphere, $-\Delta_{\mathbb{S}^{n-2}}$, it becomes apparent that the reduced estimates associated with (\ref{eq:onlyT}), (\ref{eq:hds}) and (\ref{sdst}) will follow if $f$ obeys the following two conditions:
\begin{align}
- \lambda^{min}_{\mathbb{S}^{n-2}} f \left(\frac{\xi(r)}{r^2}\right)^\prime -\left(f (V+\mu\xi)^\prime +\frac{f^{\prime\prime\prime}}{2}\right) \geq c\frac{\xi}{r^3} \label{condition1} \, , 
\end{align}
holds for a constant $c$ depending only on parameters and given such an $f$, one can find a $g$ such that also
\begin{align}\label{condition2}
-\left(2f^\prime g^2 -2(f^\prime g)^\prime -  f (V+\mu\xi)^\prime +\frac{f^{\prime\prime\prime}}{2}\right) \geq c\frac{\xi}{r^4}
\end{align}
holds. Moreover, the reduced estimate associated with (\ref{erne}) will follow if (\ref{condition1}) and (\ref{condition2}) hold with $\xi$ replaced by $\left(1-\frac{M}{r}\right)\xi$ on both right hand sides. 

\begin{remark}\label{rem: sec: proofs, rem 1}
	Note that for solutions with vanishing spherical average, $\psi_0=0$, only (\ref{condition1}) has to hold. This will be the case in the proof of Theorem~\ref{theorem 2}. 
\end{remark}

\begin{remark}
In the proof of Theorem \ref{theorem 3} we will establish (\ref{condition2}) with $\xi^2$ instead of $\xi \left(1-\frac{M}{r}\right)$, which leads to a slightly weaker estimate than (\ref{erne}), which is however easily improved a posteriori. 
\end{remark}

We complete the proof of the reduced estimate for (\ref{eq:onlyT})--(\ref{sdst}) by checking (\ref{condition1}) and (\ref{condition2}) in each case.

\subsection{Proof of Theorem~\ref{theorem 1}}\label{subsec: sec: proofs, subsec 1}

We verify~\eqref{condition1} with
\begin{equation*}
	n=4,\quad \xi =1-\frac{2M}{r},\quad \mu=0,\quad \lambda^{min}_{\mathbb{S}^2} = 2
\end{equation*}
and
\begin{equation}\label{eq: subsec: sec: proofs, subsec 1, eq 3}
	f(r)= \left(1-\frac{3M}{r}\right)\sqrt{1+\frac{6M}{r}}, \ \ \ \ \  \frac{d f}{dr} =  \frac{27 M^2}{r^3 \sqrt{1+\frac{6 M}{r}}}>0.
\end{equation}
We compute
\begin{align}\label{eq: subsec: sec: proofs, subsec 1, eq 5}
	 &  	- \lambda^{min}_{\mathbb{S}^{2}} f \left(\frac{\xi(r)}{r^2}\right)^\prime -\left(f V^\prime +\frac{f^{\prime\prime\prime}}{2}\right)=2\frac{2 \sqrt{\frac{6 M}{r}+1} (r-3 M)^2 (r-2 M)}{r^6} \nonumber \\
	 &	\qquad\qquad\qquad+ \frac{M (r-2 M) \left(-54108 M^5+20628 M^4 r+5481 M^3 r^2-1182 M^2 r^3-176 M r^4+12 r^5\right)}{2 r^8 \sqrt{\frac{6 M}{r}+1} (6 M+r)^2} \nonumber \\
&\qquad\qquad =
\frac{(r-2 M) \left(-54108 M^6+36180 M^5 r+2889 M^4 r^2-3342 M^3 r^3-104 M^2 r^4+108 M r^5+8 r^6\right)}{2 r^8 \sqrt{\frac{6 M}{r}+1} (6 M+r)^2}.
\end{align}
It is easy to check that for~$r\geq 2M$ the polynomial
\begin{equation*}
	-54108 M^6+36180 M^5 r+2889 M^4 r^2-3342 M^3 r^3-104 M^2 r^4+108 M r^5+8 r^6
\end{equation*}
is strictly positive and hence (\ref{condition1}) holds.

To verify (\ref{condition2}) we choose $g(r)=-\frac{1}{2} \left(1-\frac{2M}{r}\right) \frac{1}{r} + \frac{1}{2} \frac{M^2}{r^3}$ and compute
\begin{equation}\label{eq: subsec: sec: proofs, subsec 1, eq 7}
	2\left( (f^\prime g)^\prime - f^\prime g^2 \right) -f V^\prime - \frac{1}{2} f^{\prime \prime \prime} =
\end{equation}
{\footnotesize
\begin{equation*}
	\begin{aligned}
		\frac{M \sqrt{\frac{6 M}{r}+1} (r-2 M) \left(-972 M^7+21060 M^6 r-16551 M^5 r^2+702 M^4 r^3+1215 M^3 r^4-48 M^2 r^5+13 M r^6+12 r^7\right)}{2 r^9 (6 M+r)^3}.
	\end{aligned} \nonumber
\end{equation*}
}

It is now straightforward to conclude that the polynomial
\begin{equation*}
	p(r)=-972 M^7+21060 M^6 r-16551 M^5 r^2+702 M^4 r^3+1215 M^3 r^4-48 M^2 r^5+13 M r^6+12 r^7
\end{equation*}
satisfies
\begin{equation*}
	p(r)\geq c r^7\ \ \ \ \textrm{for $r \geq 2M$.}
\end{equation*}
for some constant~$c$. Therefore, by simple asymptotic analysis we conclude that \eqref{condition2} holds.

\subsection{Proof of Theorem~\ref{theorem 1.1}}\label{subsec: sec: proofs, subsec 2}

We verify~\eqref{condition1} with
\begin{equation*}
n=5,\quad \xi =1-\frac{2M}{r^2},\quad \mu=0,\quad \lambda^{min}_{\mathbb{S}^3} = 3
\end{equation*}
and
\begin{equation}\label{eq: proof subsec: sec: proofs, subsec 2, eq 3}
	f(r)=  \frac{1}{r^2}\left(r-2 \sqrt{M}\right) \left(2 \sqrt{M}+r\right) ,\ \  \ \ \ \frac{df}{dr}= \frac{2 \sqrt{2}}{r^3}>0. 
\end{equation}
We compute
\begin{equation}
- \lambda^{min}_{\mathbb{S}^{3}} f \left(\frac{\xi(r)}{r^2}\right)^\prime -\left(f V^\prime +\frac{f^{\prime\prime\prime}}{2}\right)= \frac{\frac{1}{ r^2} \left(r^2-2M\right) \left(-848 M^3 + 756 M^2 r^2 - 180 M r^4 + 15 r^6\right)}{4 \sqrt{2} r^9}.
\end{equation}
It is easy to check that the polynomial
\begin{equation*}
	-848 M^3 + 756 M^2 r^2 - 180 M r^4 + 15 r^6
\end{equation*}
is strictly positive for~$r>\sqrt{2M}$ by rescaling~$M=1$ and by proving that 
\begin{equation*}
	15 r^6-180 r^4+756 r^2-848\geq 64
\end{equation*}
for~$r>\sqrt{2}$. We have concluded~\eqref{condition1}. 

To verify~\eqref{condition2} we choose~$g(r)=\frac{1}{2 r^3}-\frac{1-\frac{2 M}{r^2}}{2 r}$ and compute 
\begin{equation}
	 2\left( (f^\prime g)^\prime - f^\prime g^2 \right)-fV^\prime -\frac{1}{2}f^{\prime\prime\prime} = \frac{\sqrt{\frac{1}{M r^4}} \left(r^2-2M\right) \left(-152 M^3+132 M^2 r^2-28 M r^4+3 r^6\right)}{4 \sqrt{2} r^9}.
\end{equation}
It is now straightforward to conclude that the polynomial
\begin{equation}
	p(r)=-152 M^3+132 M^2 r^2-28 M r^4+3 r^6
\end{equation}
satisfies 
\begin{equation*}
 p(r)\geq c \cdot r^6,\qquad r \geq \sqrt{2M}
\end{equation*}
for some constant $c>0$. Therefore, by simple asymptotic analysis we conclude that \eqref{condition2} holds.

\subsection{Proof of Theorem~\ref{theorem 3}}\label{sec: proof of theorem 3}

We verify~\eqref{condition1} with
\begin{equation}
n=4,\quad \xi =\left(1-\frac{M}{r}\right)^2,\quad \mu=0 ,\quad \lambda^{min}_{\mathbb{S}^2} = 2
\end{equation}
and
\begin{equation}\label{eq: subsec: sec: proofs, subsec 6, eq 1.1}
	f(r)=\frac{(r-2 M) \sqrt{r (4 M+r)-4 M^2}}{r^2}, \ \ \ \ \ \frac{d f}{dr} = \frac{(4 M)^2 (r-M)}{r^3 \sqrt{-4 M^2+4 M r+r^2}}\geq 0 .
\end{equation}
We compute 
\begin{equation*}
	\begin{aligned}
		- \lambda^{min}_{\mathbb{S}^{2}} f \left(\frac{\xi(r)}{r^2}\right)^\prime -\left(f V^\prime +\frac{f^{\prime\prime\prime}}{2}\right) &	=
		\frac{2M^6(r-M)^3}{M^3 r^{11} \sqrt{-4 M^2+2 r (2 M+r)-r^2} \left(8 M^4+2 M^2 r^2-4 M^2 r (2 M+r)\right)^2} \times \\
		&	\times \Bigg(-1792 M^{10}+8960 M^9 r-17920 M^8 r^2+17920 M^7 r^3-8600 M^6 r^4+712 M^5 r^5\\
		&	\qquad\qquad  +1056 M^4 r^6-312 M^3 r^7-43 M^2 r^8+19 M r^9+2 r^{10}\Bigg).
	\end{aligned}
\end{equation*}
It is easy to check that the polynomial
\begin{equation*}
	\begin{aligned}
		&-1792 M^{10}+8960 M^9 r-17920 M^8 r^2+17920 M^7 r^3-8600 M^6 r^4+712 M^5 r^5+1056 M^4 r^6\\
		&	\qquad -312 M^3 r^7-43 M^2 r^8+19 M r^9+2 r^{10}.
	\end{aligned}
\end{equation*}
is strictly positive for~$r\geq M$ by rescaling~$M=1$ and by proving that
\begin{equation*}
	2 r^{10}+19 r^9-43 r^8-312 r^7+1056 r^6+712 r^5-8600 r^4+17920 r^3-17920 r^2+8960 r-1792\geq 2
\end{equation*}
for~$r\geq 1$. We have concluded~\eqref{condition1}.

To verify \eqref{condition2} we choose~$g(r)= -\frac{1}{2r}\left(1-\frac{M}{r}\right)^2+\frac{1}{2}\left(1-\frac{M}{r}\right)\frac{M^2}{r^3}$ 
and compute
\begin{equation}
 2\left( (f^\prime g)^\prime - f^\prime g^2 \right)-fV^\prime -\frac{1}{2}f^{\prime\prime\prime}
=
\frac{(r-M)^4}{2 r^{13} \left(-4 M^2+4 M r+r^2\right)^{5/2}}p(r) \, ,
\end{equation}
where
\begin{equation}
	\begin{aligned}
		p(r)&	=64 M^{10}-1216 M^9 r+4640 M^8 r^2-7200 M^7 r^3+4788 M^6 r^4\\
		&	\qquad\qquad -732 M^5 r^5-444 M^4 r^6+92 M^3 r^7+4 M^2 r^8+4 M r^9+3 r^{10}.
	\end{aligned}
\end{equation}
It is now straightforward to conclude that the polynomial~$p(r)$ satisfies 
\begin{equation}
	p(r) \geq c \cdot r^{10},\qquad r \geq M
\end{equation}
for some $c>0$.

We concluded that the estimate \eqref{condition2} holds with $\xi^2$ instead of $\xi \left(1-\frac{M}{r}\right)$ on the right hand side and leads to \eqref{erne} except that the zeroth order term has an additional $(1-\frac{M}{r})$-degeneration at the horizon. Another Hardy type inequality using the good $(\partial_{r^\star} \psi)^2$-term easily removes this additional degeneration a posteriori and the proof is complete.

\subsection{Proof of Theorem~\ref{theorem 2}}\label{subsec: sec: proofs, subsec 3}	

Note that in the present proof we only verify~\eqref{condition1}, but not condition~\eqref{condition2}, see the relevant Remark~\ref{rem: sec: proofs, rem 1}. 

We verify~\eqref{condition1} with 
\begin{equation*}
	n=4,\quad \xi =1-\frac{2M}{r}-\frac{r^2}{L^2},\quad \mu=\frac{2}{L^2},\quad \lambda^{min}_{\mathbb{S}^2} = 2 ,
\end{equation*}
where~$L^2=\frac{3}{\Lambda}$ and
	\begin{equation*}
		f(r)= \frac{1}{\sqrt{1-27\frac{M^2}{L^2}}}\left(1-\frac{3M}{r}\right)\sqrt{1+\frac{6M}{r}} , \ \  \  \ \ \frac{d f}{dr} =  \frac{1}{\sqrt{1-27\frac{M^2}{L^2}}}\frac{27 M^2}{r^3 \sqrt{1+\frac{6 M}{r}}}>0.
	\end{equation*}
We compute 
	\begin{equation*}
	- \lambda^{min}_{\mathbb{S}^{2}} f \left(\frac{\xi(r)}{r^2}\right)^\prime -\left(f (V+\mu\xi)^\prime +\frac{f^{\prime\prime\prime}}{2}\right) = 	\frac{1}{2r^7(r+6M)^2\sqrt{1-27\frac{M^2}{L^2}}\sqrt{1+\frac{6M}{r}}} \xi S(r),
\end{equation*}
	where 
	\begin{equation}\label{eq: proof theorem 2, eq 2.1}
		\begin{aligned}
			S(r)	&	= -\frac{4 M \left(4212 M^4+702 M^3 r-189 M^2 r^2-39 M r^3+r^4\right)r^3}{L^2}\\
			&	\qquad +\frac{81 M^3 (3 M+2 r) r^6}{L^4}\\
			&	\qquad -54108 M^6+28404 M^5 r+4185 M^4 r^2-2262 M^3 r^3-140 M^2 r^4+60 M r^5+4 r^6.
		\end{aligned}
	\end{equation}

We note that 
\begin{equation*}
	0<\frac{M^2}{L^2}<\frac{1}{27},\qquad 2M<r_+<r<L  ,
\end{equation*}
where the last inequality is immediate from that~$\xi(2M)<0$,~$\xi(L)<0$.

Moreover, the polynomial  
	\begin{equation*}
		-54108 M^6+28404 M^5 r+4185 M^4 r^2-2262 M^3 r^3-140 M^2 r^4+60 M r^5+4 r^6
	\end{equation*}
	is strictly increasing and specifically
	\begin{equation*}
		-54108 M^6+28404 M^5 r+4185 M^4 r^2-2262 M^3 r^3-140 M^2 r^4+60 M r^5+4 r^6\geq 1280M^6,
	\end{equation*}
	where~$1280 M^6$ is its value at~$r=2M$. For simplicity, we denote
	\begin{equation*}
		1-\frac{2M}{r}-\frac{r^2}{L^2}=\xi(r) >0 \implies r^3 = L^2 (r-2M-\xi(r) r) \, .
	\end{equation*}
	We rewrite~$S(r)$ as follows 
	\begin{equation}\label{eq: proof theorem 2, eq 3}
	\begin{aligned}
		S(r)	&	=	\xi^2(r) \left(243 M^4 r^2+162 M^3 r^3\right)\\
		&	\qquad +\xi(r) \left(17820 M^5 r+2970 M^4 r^2-1080 M^3 r^3-156 M^2 r^4+4 M r^5\right)\\
		&	\qquad  + 4 (3 M-r)^2 (6 M+r)^2 \left(r^2-15 M^2+8 M r\right),
	\end{aligned}
\end{equation}
where we note that~$r^2+8Mr-15M^2>0$ for~$r\in [r_+,\bar{r}_+]$.

Now, suppose that there exists~$r_0\in [r_+,\bar{r}_+]$ such that
\begin{equation}\label{eq: proof theorem 2, eq 4}
	 \left(17820 M^5 r+2970 M^4 r_0^2-1080 M^3 r_0^3-156 M^2 r_0^4+4 M r_0^5\right)<0,
\end{equation}
otherwise~$S(r)>0$. Then, by using that~$\xi(r_0)<1$ we bound~$S(r_0)$ from below as follows 
\begin{equation*}
	\begin{aligned}
		S(r_0)	&	>  \left(17820 M^5 r+2970 M^4 r_0^2-1080 M^3 r_0^3-156 M^2 r_0^4+4 M r_0^5\right)\\
		&	\qquad  + 4 (3 M-r_0)^2 (6 M+r_0)^2 \left(r_0^2-15 M^2+8 M r_0\right)\\
		&	=	-19440 M^6+34668 M^5 r_0+2430 M^4 r_0^2-2736 M^3 r_0^3-132 M^2 r_0^4+60 M r_0^5+4 r_0^6.
	\end{aligned}
\end{equation*}
A direct critical point analysis of the polynomial
\begin{equation*}
	-19440 M^6+34668 M^5 r+2430 M^4 r^2-2736 M^3 r^3-132 M^2 r^4+60 M r^5+4 r^6
\end{equation*}	
proves that for~$r>2M$ we obtain that 
\begin{equation*}
		-19440 M^6+34668 M^5 r+2430 M^4 r^2-2736 M^3 r^3-132 M^2 r^4+60 M r^5+4 r^6>20000M^6,
\end{equation*}	
	which proves that even for~$r_0$ values such that~\eqref{eq: proof theorem 2, eq 4} holds we obtain that~$S(r_0)>0$. 
	
We have concluded~\eqref{condition1} and the proof.

\section{Normal hyperbolicity and commutator vector fields} \label{sec:nohyp}

In this section, we elucidate the connection between the multipliers of the present paper and the globally good commutators introduced in~\cite{gustav,mavrogiannis}. For this purpose, we continue our discussion on normal hyperbolicity from Section \ref{sec:constructionprinciple} using the language of dynamical systems, but restricting now to the Schwarzschild spacetime. This is merely for simplicity of the expressions. Identical considerations apply -- with trivial algebraic changes -- to any of the spacetimes $(1)$--$(4)$. In this framework, we will show that the globally good commutators introduced in~\cite{gustav,mavrogiannis} arise by studying the expansion and contraction properties of the radial geodesic flow.

\subsection{Characterisation of the future/past-trapped sets}

 We first summarise our discussion in Section \ref{sec:constructionprinciple} by the following proposition.

\begin{proposition} \label{prop:gf}
	The sets of future-trapped and past-trapped geodesics in Schwarzschild are analytic codimension one submanifolds of phase space $\mathcal{P}$ given by
	\begin{equation*}
		W^+=\Big\{(x,p)\in  \mathcal{P}:\dfrac{l}{E}=3\sqrt{3}M,~ \frac{p^{r}}{E}=\Big(1+\dfrac{6M}{r}\Big)^{\frac{1}{2}}\Big(1-\dfrac{3M}{r}\Big) \Big\},
		\end{equation*}
		and
		\begin{equation*}
		W^-=\Big\{ (x,p)\in \mathcal{P}:\dfrac{l}{E}=3\sqrt{3}M,~  \frac{p^{r}}{E}=\Big(1+\dfrac{6M}{r}\Big)^{\frac{1}{2}}\Big(\dfrac{3M}{r}-1\Big)\Big\}.\end{equation*}
respectively. In particular, the intersection $W^+\cap W^-$ is equal to the trapped set $\Gamma$.	
\end{proposition}

We will now write the sets $W^{\pm}$ of future-trapped and past-trapped geodesics in terms of specific defining functions $\varphi_{\pm}$. 
Motivated by \eqref{quotient_conserved_quantities}, we set the functions $\varphi_{\pm}\colon \mathcal{P}\to \mathbb{R}$ given by
\begin{equation*}
\varphi_{\pm}(x,p):=\frac{r^\frac{3}{2}}{(r-2M)^{\frac{1}{2}}}\Big(1+\frac{6M}{r}\Big)^{\frac{1}{2}}\Big(1-\frac{3M}{r}\Big)\pm\frac{r^\frac{3}{2}}{(r-2M)^{\frac{1}{2}}}\Big(\frac{p^r}{E}\Big).
\end{equation*}
Note that $\frac{l^2}{E^2}-27M^2=\varphi_{+}\varphi_{-}$ by \eqref{quotient_conserved_quantities}. In terms of the defining functions $\varphi_{\pm}$, the sets of future-trapped and past-trapped geodesics $W^{\pm}$ can then be written as 
\begin{equation*}
W^+=\Big\{ (x,p)\in \mathcal{P}:\varphi_{-}(x,p)=0\Big\},\qquad W^-=\Big\{ (x,p)\in \mathcal{P}:\varphi_{+}(x,p)=0\Big\}.
\end{equation*}

\subsection{Normal hyperbolicity of the trapped set}

The trapped set on Schwarzschild is \emph{eventually absolutely $r$-normally hyperbolic} for every $r$. See \cite[Section 1, Definition 4]{HPS77} for a precise definition. This property was first proven in \cite{WZ11}. In terms of the defining functions $\varphi_{\pm}$ of the stable manifolds $W^{\pm}$, the normal hyperbolicity of the geodesic flow is expressed in the following proposition.

\begin{proposition}[Expansion/contraction of the radial flow]\label{prop_normal_hyp}
The derivative of $\varphi_{\pm}$ along the geodesic flow in Schwarzschild is 
\begin{equation}\label{exp_contr_flow}
\frac{d}{ds}\varphi_{\pm}(x,p)=\pm \frac{1}{(1-\frac{2M}{r})r(1+\frac{6M}{r})^{\frac{1}{2}}}E\varphi_{\pm}(x,p).
\end{equation}
\end{proposition}

Proposition \ref{prop_normal_hyp} follows by elementary calculations using the radial geodesic equation \eqref{radgeoeqn}.\\

The radial geodesic flow is hyperbolic in any bounded region by the expansion/contraction properties in \eqref{exp_contr_flow}. In particular, the hyperbolicity of the radial flow induces suitable stable and unstable invariant distributions on $\mathcal{P}$.\footnote{Recall that a \emph{distribution} in phase space $\mathcal{P}$ is a map $(x,p)\mapsto \Delta_{(x,p)}\subseteq T_{(x,p)}\mathcal{P}$ where $\Delta_{(x,p)}$ are vector subspaces satisfying suitable regularity conditions.} 
Here, by an invariant distribution we mean a distribution invariant by the action of the differential of the geodesic flow in $T\mathcal{P}$. For the study of the differential of the geodesic flow, we consider the normals $V_{\pm}$ of the sets $W^{\pm}$. Set the vector fields 
\begin{equation}\label{vector_field_future_trap}
	V_+:=\dfrac{r^{\frac{3}{2}}}{(r-2M)^{\frac{1}{2}}}\Big(\Big(1+\dfrac{6M}{r}\Big)^{\frac{1}{2}}\Big(\dfrac{3M}{r}-1\Big)\partial_t+\partial_{r_*}\Big)\quad \textrm{such that} \quad g(V_+,p)=\varphi_-,
\end{equation}
and
\begin{equation}\label{vector_field_past_trap}
	V_-:=\dfrac{r^{\frac{3}{2}}}{(r-2M)^{\frac{1}{2}}}\Big(\Big(1+\dfrac{6M}{r}\Big)^{\frac{1}{2}}\Big(1-\dfrac{3M}{r}\Big)\partial_t+\partial_{r_*}\Big)\quad \textrm{such that} \quad g(V_-,p)=\varphi_+,
\end{equation}
on the stable manifolds $W^{\pm}$, respectively. By considering the vector fields \eqref{vector_field_future_trap}--\eqref{vector_field_past_trap} in a suitable moving frame, one can show the expansion/contraction properties of the differential of the flow map on suitable distributions of $\mathcal{P}$. A detailed discussion about the differential of the geodesic flow is beyond the scope of this paper. See \cite{V23} for further details.

\begin{remark}
Our discussion has so far been restricted to the geodesic flow in Schwarzschild spacetime where all sets and vectorfields can be parametrised explicitly. However, the above concepts apply in much greater generality. By the stable manifold theorem for normally hyperbolic sets by Hirsch, Pugh, and Shub \cite{HPS77}, a generalisation of the Hadamard--Perron theorem, the normal hyperbolicity of a trapped set $\Gamma$ implies the existence of suitable stable and unstable manifolds $W^\pm$ (locally given as the zero set of functions $\varphi_\pm$) in phase space $\mathcal{P}$. The vectorfields $V_\pm$ are then defined abstractly by $g(V_\pm,p)=\varphi_\mp$. 
\end{remark}

\subsection{Commutator vector fields for the wave operator}

Remarkably, the vector field \eqref{vector_field_past_trap} satisfies good commuting properties with the wave operator $\Box_{g} $ on Schwarzschild spacetime as shown in \cite{gustav,mavrogiannis}. The commutation of $V_-$ with the wave operator, corresponds in the high frequency regime to the Poisson bracket with the defining function $\varphi_{+}$ obtained in Proposition \ref{prop_normal_hyp}. For comparison, we state the commutation formulae with the wave operator obtained in \cite[equation (12)]{gustav} and \cite[Proposition 4.1]{mavrogiannis}.

\begin{proposition}\label{propcommSch}
	Let $\phi$ be a solution to the wave equation on a Schwarzschild black hole~$\mathrm{Schw}^{1+3}$. Then, the following holds
	\begin{equation}\label{eq: equation obeyed by Psi}
		\Box (V_-\phi)=\big[\Box,V_-\big]\phi=\frac{2}{(1-\frac{2M}{r})r(1+\frac{6M}{r})^{\frac{1}{2}}}\partial_tV_-\phi+ E_1(r)\partial_t\phi + E_2(r)\frac{1}{1-\frac{2M}{r}}\left(\partial_{r^\star}+f(r)\partial_t\right)\phi ,
	\end{equation}
	where 
	\begin{equation}
		\begin{aligned}
			E_1(r)  =-\frac{2(1-\frac{2M}{r})^{\frac{1}{2}}}{r}, \qquad 
			E_2(r)  = \frac{M^2}{r^3(1-\frac{2M}{r})^{\frac{1}{2}}},\qquad f(r)= \left(1-\frac{3M}{r}\right)\Big(1+\frac{6M}{r}\Big)^{\frac{1}{2}}.
		\end{aligned}
	\end{equation}
	\end{proposition}

We recall that~\cite{gustav} obtains local integrated decay estimates for the perturbed wave equation $\Box_{g}\phi =\epsilon \beta^{a}\partial_a \phi$ on Schwarzschild by using Proposition \ref{propcommSch}. Here $\beta$ is a regular vector field suitably decaying in space. On the other hand, \cite{mavrogiannis} proves relatively non-degenerate integrated energy estimates for the wave equation on subextremal Schwarzschild--de Sitter spacetimes by using a suitable modification of Proposition \ref{propcommSch}. Extensions of Proposition \ref{propcommSch} have been obtained and exploited for the study of local integrated decay estimates on Kerr and Kerr--de Sitter spacetimes in \cite{mav2, gustav2}.

\section{Final remarks}
We finish the paper with a few remarks about local integrated decay estimates on other spherically symmetric black hole spacetimes.

\begin{remark}[On higher dimensional Schwarzschild]
	Let
	\begin{equation*}
		g = -\Big(1-\frac{2M}{r^{n-3}}\Big)dt \otimes dt +\Big(1-\frac{2M}{r^{n-3}}\Big)^{-1}dr \otimes dr +r^2 d\sigma_{\mathbb{S}^{n-2}},
	\end{equation*}
	where~$n\geq 4$, be the metric of an~$n$ dimensional Schwarzschild black hole. For more information about this specific geometry see~\cite{volker-HighSchwarzschild}. Motivated by the considerations of Section~\ref{sec:constructionprinciple}, we define the multiplier 
	\begin{equation*}
		\begin{aligned}
			f(r)=\frac{r-	 r_{trap}}{r^{\frac{n-1}{2}}}\sqrt{ \sum_{k=2}^{n-1} \frac{-(r-	 r_{trap})^{k-2}}{k!} \partial_{r}^k\Big|_{r=	 r_{trap}}\left(-(	 r_{trap})^{-2} r^{n-1}+2 M \left(	 r_{trap}\right)^{1-n} r^{n-1}-2 M+r^{n-3}\right)},
		\end{aligned}
	\end{equation*}
	where~$		 r_{trap} =(n-1)^{\frac{1}{n-3}} M^{\frac{1}{n-3}}$ is the value that defines the corresponding photon sphere. We expect that by using this $f$ in (\ref{figo}), one produces a coercive estimate also for $n\geq 6$. 
\end{remark}

\begin{remark}[On subextremal Reissner--Nordstr\"om]\label{rem: subsec: sec: main thms, subsec 3, rem 1}
Let
\begin{equation*}
	g =-\left(1-\frac{2M}{r}+\frac{Q^2}{r^2}\right)dt \otimes dt +\left(1-\frac{2M}{r}+\frac{Q^2}{r^2}\right)^{-1} dr \otimes dr  + r^2 d\sigma_{\mathbb{S}^2}
\end{equation*}
be the metric of a subextremal Reissner--Nordstr\"om black hole~$|Q|<M$. Motivated by the considerations of Section~\ref{sec:constructionprinciple}, we define the multiplier 
\begin{equation}\label{eq: rem: subsec: sec: main thms, subsec 3, rem 1, eq 1}
	\begin{aligned}
		f(r)		= \frac{1}{r}\Big(1-\dfrac{ r_{trap}}{r}\Big)\Big(r^2+2 r_{trap} r-\dfrac{ r_{trap}^2Q^2}{\Delta( r_{trap})}\Big)^{\frac{1}{2}}, \qquad \Delta(r)=r^2-2Mr+Q^2,
	\end{aligned}
\end{equation}
where
\begin{equation*}
		 r_{trap}= \frac{1}{2} \left( \sqrt{9M^2-8Q^2}+3M\right)
\end{equation*}
is the value that defines the corresponding photon sphere. The authors have verified for several subextremal parameters that the multiplier above produces coercive estimates for solutions with vanishing spherical mean. With more effort, one should be able to check the entire subextremal range. 
\end{remark}

\begin{remark}[On subextremal Reissner--Nordstr\"om--de~Sitter]\label{rem: subsec: sec: main thms, subsec 3, rem 2}
	Let
	\begin{equation*}
		g =-\left(1-\frac{2M}{r}+\frac{Q^2}{r^2}-\frac{r^2}{L^2}\right)dt\otimes dt +\left(1-\frac{2M}{r}+\frac{Q^2}{r^2}-\frac{r^2}{L^2}\right)^{-1} dr\otimes dr + r^2 d\sigma_{\mathbb{S}^2}
	\end{equation*}
	be the metric of a Reissner--Nordstr\"om--de~Sitter black hole, where~$L=\frac{3}{\Lambda}$ is defined in terms of the cosmological constant $\Lambda> 0$. To the best of our knowledge a Morawetz estimate for the wave equation on Reissner--Nordstr\"om--de~Sitter~$(\Lambda>0)$ has not been proved in the literature. It would be interesting to understand further the Morawetz estimate produced by the multiplier suggested from the considerations of Section~\ref{sec:constructionprinciple}, in a near extremal Reissner--Nordstr\"om--de Sitter black hole, where a violation of Strong Cosmic Censorship has been conjectured~\cite{Cardoso-Costa-Destounis-RNdS}. Specifically, it would be interesting to track down the constants in the Morawetz estimate and the relevant commutator that gives a `relatively non-degenerate estimate'. See~\cite{mavrogiannis} for the proof of a `relatively non-degenerate estimate' on Schwarzschild--de~Sitter. 
\end{remark}

\addtocontents{toc}{\protect\setcounter{tocdepth}{0}}

\addtocontents{toc}{\protect\setcounter{tocdepth}{1}}

\bibliographystyle{plain}
\bibliography{MyBibliography}

\end{document}